\documentclass[final,twoside,11pt]{entics} 
\usepackage{enticsmacro}
\usepackage{amsmath, amssymb}
\usepackage{graphicx}
\usepackage{relsize,amsfonts}
\usepackage{framed}
\usepackage{xcolor}
\usepackage[all]{xy}
\newcommand{\lar}[1]{\mathlarger{\mathlarger{\mathlarger{#1}}}}
\newcommand{\J}{\mathsf{J}}
\newcommand{\K}{\mathsf{K}}

\def\conf{MFPS 2023} 	
\def\myconf{mfps39}
\volume{3}			
\def\volu{3}			
\def\papno{16}			
\def\lastname{Maschio \& Miquel} 
\def\runtitle{Implicative models} 
\def\mydoi{12426}
\def\copynames{S.Maschio, A. Miquel}

\begin{document}
\begin{frontmatter}
  \title{Implicative Models of Set Theory}
  \author{Samuele Maschio\thanksref{a}\thanksref{myemail}}
  \author{Alexandre Miquel\thanksref{b}\thanksref{myemail2}}	 
  \address[a]{Dipartimento di Matematica ``Tullio Levi-Civita''\\
    Universit\`a di Padova, Padova, Italy}  
  \address[b]{Instituto de Matemática y Estadística ``Rafael Laguardia'',\\
    Facultad de Ingeniería, Universidad de la República, Montevideo, Uruguay.}
   \thanks[myemail]{Email: \href{maschio@math.unipd.it}
     {\texttt{\normalshape maschio@math.unipd.it}}} 
   \thanks[myemail2]{Email: \href{amiquel@fing.edu.uy}
     {\texttt{\normalshape amiquel@fing.edu.uy}}} 
\begin{abstract} In this paper we show that using implicative algebras one can produce models of set theory generalizing Heyting/Boolean-valued models and realizability models of $\mathbf{(I)ZF}$, both in intuitionistic and classical logic. This has as consequence that any topos which is obtained from a $\mathbf{Set}$-based tripos as the result of the tripos-to-topos construction hosts a model of intuitionistic or classical set theory, provided a
large enough strongly inaccessible cardinal exists.
\end{abstract}
\end{frontmatter}

\section{Introduction} Implicative algebras were introduced by the second author \cite{miq1} in order to provide a common foundation for the model-theoretic constructions underlying forcing and realizability. 

Forcing was first introduced by Cohen \cite{COH} in the 60's to prove the independence of the Continuum Hypothesis with respect to $\mathbf{ZFC}$ and it is the main technique used in set theory to obtain relative consistency results. From an algebraic point of view, the technique of forcing amounts to the construction of a Boolean-valued model of the considered theory, and this construction can be further generalized to intuitionistic theories by considering Heyting-valued models. On the other hand, realizability has been introduced by Kleene \cite{KLE} in the 40's to interpret the computational content of intuitionistic proofs. For a long time, this technique was limited to intuitionistic theories, but in the mid-90's Krivine \cite{KRI} showed how to reformulate its very principles to make them compatible with classical logic.

In order to compare forcing and realizability, Hyland, Johnstone and Pitts introduced the notion of tripos \cite{HJP}, that can be seen as a categorical model of higher-order logic. 
Such triposes may be constructed from complete Heyting algebras, thus yielding forcing triposes, or from partial combinatory algebras, thus yielding intuitionistic realizability triposes. More recently, Streicher \cite{STR13} showed how to turn an abstract Krivine structure into a tripos, thus describing classical realizability in categorical terms. All these triposes $\mathsf{P}$ can then be turned into toposes $\mathbf{Set}[\mathsf{P}]$ by applying the tripos-to-topos construction \cite{HJP}.

In \cite{miq1} the second author showed that forcing and realizability triposes are instances of a more general notion of tripos induced by an implicative algebra (which is called \emph{implicative tripos}). Later, in \cite{miq2}, he proved that every $\mathbf{Set}$-based tripos is in fact (isomorphic to) an implicative tripos. These triposes include those arising from many different variants of realizability, such as modified realizability triposes, relative realizability triposes, classical realizability triposes and so on.

However, it is worth recalling that, from a proof-theoretic perspective, topos theory is much weaker than set theory. Indeed, the internal theory of a topos with a natural numbers object is strictly weaker than (intuitionistic or classical) Zermelo set theory (IZ/Z), that is itself much weaker than (intuitionistic or classical) Zermelo-Fraenkel set-theory (IZF/ZF).
Intuitively, the main difference is that in topos theory, one can only quantify over the elements of a given set or object (i.e.\ bounded quantification), whereas in set theory, one can also quantify over all sets/objects (i.e.\ unbounded quantification).

In this paper we shall see how implicative algebras can be used to construct \emph{implicative} models of intuitionistic and classical set theory (depending on whether the underlying implicative algebra is intuitionistic or classical).
We will then see that our implicative models of set theory encompass Heyting/Boolean-valued models for $\mathbf{(I)ZF}$ \cite{bell,BELL2} and, up to logical equivalence, Friedman/Rosolini/McCarty realizability models for $\mathbf{IZF}$ \cite{Friedman,Rosolini,McCarty} as well as the classical realizability models of $\mathbf{ZF}$ introduced by Krivine \cite{KRIZF1,KRIZF2}.

Finally, in the last section, we will use the relationship between the logic of a tripos $\mathsf{P}$ and the internal logic of the corresponding topos $\mathbf{Set}[\mathsf{P}]$ to show that our implicative models of $\mathbf{(I)ZF}$ can be seen as internal models of set theory in the toposes constructed from
   implicative triposes. Since the second author proved \cite{miq2} that every $\mathbf{Set}$-based tripos is an implicative tripos (up to isomorphism), we can conclude that every topos induced by a $\mathbf{Set}$-based tripos hosts an internal model of set theory (provided a large enough cardinal exists).

\section{Intuitionistic and classical set theory}
In the language of Zermelo-Fraenkel set theory the only terms are variables and there are two binary predicate symbols: equality $=$ and membership $\in$.
As usual in the language of set theory $\forall x\in y\,\varphi$ is a shorthand for $\forall x(x\in y\rightarrow \varphi)$ and $\exists x\in y\,\varphi$ is a shorthand for $\exists x(x\in y\wedge \varphi)$, while $x\subseteq y$ is a shorthand for $\forall z\in x\,(z\in y)$.

The theories $\mathbf{ZF}$ and $\mathbf{IZF}$ have both the following axioms, but the logic underlying the former is classical, while it is intuitionistic for the latter:
\begin{enumerate}
\item[$\mathbf{Ext})$] $\forall x\forall y\,(x\subseteq y\wedge y\subseteq x\rightarrow x=y)$ 
\item[$\mathbf{Pair})$] $\forall x\forall y\exists z\,(x\in z\wedge y\in z)$
\item[$\mathbf{Union})$] $\forall x\exists u\forall y\in x\forall z\in y\,(z\in u)$
\item[$\mathbf{Pow})$] $\forall x\exists z\forall y\,(y\subseteq x\rightarrow y\in z)$
\item[$\mathbf{Inf})$] $\exists u\mathbf{Inf}(u)$ where $\mathbf{Inf}(u)$ is the conjunction of the formulas $\mathbf{Inf}_1(u):\equiv\exists x\in u\forall y\in x\,\bot$ and $\mathbf{Inf}_2(u):\equiv\forall x\in u\exists y\in u(x\subseteq y\wedge x\in y\wedge \forall z\in y(z\in x \vee z=x))$
\item[$\mathbf{Sep}_{\varphi})$] $\forall w_1....\forall w_n\forall x\exists y\left(\forall z\in y\,(z\in x\wedge \varphi)\wedge \forall z\in x\,(\varphi\rightarrow z\in y)\right)$ for all formulas $\varphi[w_1,...w_n,x,z]$  in context.
\item[$\mathbf{Col}_\varphi)$] $\forall w_1....\forall w_n\forall y\left(\forall x\in y \exists z\, \varphi \rightarrow \exists u\forall x\in y\exists z\in u\,\varphi\right)$ for all formulas in context $\varphi[w_1,...w_n,x,y,z]$.
\item[$\mathbf{\in\textrm{-}Ind}_{\varphi})$] $\forall w_1...\forall w_n(\forall x (\forall y\in x\,\varphi[y/x]\rightarrow \varphi)\rightarrow \forall x\,\varphi)$ for all formulas in context $\varphi[w_1,...,w_n,x]$.
\end{enumerate}

\section{Implicative algebras and implicative triposes}\label{s:ImpAlgTrip}
An \emph{implicative algebra} is a $4$-tuple $\mathbb{A}=(A,\leq,\rightarrow,\Sigma)$ where 
\begin{enumerate}
\item $(A,\leq)$ is a complete lattice;
\item $\rightarrow:A\times A\rightarrow A$ is a function which is monotone in the second component and anti-monotone in the first component, and which satisfies the following condition 
$$a\rightarrow \bigwedge_{i\in I}b_i=\bigwedge_{i\in I}\left(a\rightarrow b_i
\right)$$
for every indexed family $(b_i)_{i\in I}$ of elements of $A$ and every $a\in A$;
\item $\Sigma\subseteq A$  is upward closed, it contains $b$ as soon as it contains $a\rightarrow b$ and $a$, and it contains $\mathbf{K}:=\bigwedge_{a,b\in A}(a\rightarrow (b\rightarrow a))$ and $\mathbf{S}:=\bigwedge_{a,b\in A}((a\rightarrow (b\rightarrow c))\rightarrow ((a\rightarrow b)\rightarrow (a\rightarrow c)))$.
\end{enumerate}

Every complete Heyting algebra $(H,\leq)$ with Heyting implication $\rightarrow$ gives rise to an implicative algebra $(H,\leq,\rightarrow,\{\top\})$. Moreover, every total combinatory algebra $(R,\cdot)$ gives rise to an implicative algebra $(\mathcal{P}(R),\subseteq,\Rightarrow, \mathcal{P}(R)\setminus \{\emptyset\})$, where $A\Rightarrow B:=\{r\in R|\; r\cdot a \in B \textrm{ for every }a\in A\}$ for every $A,B\subseteq R$. Other examples can be found in \cite{miq1}. In the case of a partial combinatory algebra $(R,\cdot)$, the $4$-uple $(\mathcal{P}(R),\subseteq,\Rightarrow, \mathcal{P}(R)\setminus \{\emptyset\})$ is not in general an implicative algebra, but a quasi-implicative algebra (see \cite{miq1}). However, there is a standard way to transform it into an implicative algebra in such a way that the tripos one obtains from it is equivalent to the realizability tripos built from $(R,\cdot)$ (for details, see \cite{miq1}).

An implicative algebra $\mathbb{A}=(A,\leq,\rightarrow,\Sigma)$ is \emph{classical} if $\bigwedge_{a,b\in A}(((a\rightarrow b)\rightarrow a)\rightarrow a)\in \Sigma$.
Complete Boolean algebras give rise to classical implicative algebras following the recipe used for complete Heyting algebras.
Last but not least, classical implicative algebras can also be constructed from Abstract Krivine Structures~\cite{STR13}, the algebraic structure underlying classical realizability~\cite{KRI11}.

Closed $\lambda$-terms with constant parameters in an implicative algebra $\mathbb{A}$ can be encoded as elements of $\mathbb{A}$ itself as follows: $a^{\mathbb{A}}:=a$ for every $a\in A$, $(ts)^{\mathbb{A}}:=t^{\mathbb{A}}\cdot s^{\mathbb{A}}$ and $(\lambda x.t)^{\mathbb{A}}:=\bigwedge_{a\in A}\left(a\rightarrow (t[a/x])^{\mathbb{A}}\right)$ where the application $\cdot$ is defined as follows for every $a,b\in A$:
$$a\cdot b:=\bigwedge\{x\in A|\,a\leq b\rightarrow x\}$$
If we define the combinators $\mathbf{k}$ as $\lambda x.\lambda y.x$ and $\mathbf{s}$ as $\lambda x.\lambda y.\lambda z.xz(yz)$ as usual, one can show (see \cite{miq1}) that $\mathbf{K}=\mathbf{k}^{\mathbb{A}}$ and $\mathbf{S}=\mathbf{s}^{\mathbb{A}}$.

Useful properties of the encoding of $\lambda$-terms in $\mathbb{A}$ are the following:
\begin{enumerate}
\item if $t$ $\beta$-reduces to $s$, then $t^{\mathbb{A}}\leq s^{\mathbb{A}}$;
\item if $t$ is a pure $\lambda$-term with free variables $x_1,...,x_n$ and $a_1,...,a_n\in \Sigma$, then $(t[x_1:a_1,...,x_n:a_n])^{\mathbb{A}}\in \Sigma$\footnote{We denote with $t[x_1:a_1,...,x_n:a_n]$ the $\lambda$-term obtained from $t$ by substituting the variables $x_1,...,x_n$ with $a_1,...,a_n$, respectively.}; in particular the encodings of closed pure $\lambda$-terms are elements of $\Sigma$.
\end{enumerate}
In what follows we will remove the superscript $\mathbb{A}$ from the encoding of $\lambda$-terms in order to lighten the notation.

For $a,b\in A$, we write $a\vdash_{\Sigma}b$ if $a\rightarrow b\in \Sigma$, while we write $a\equiv_{\Sigma}b$ if $a\vdash_{\Sigma}b$ and $b\vdash_{\Sigma}a$, moreover for every $a,b\in A$ following \cite{miq1} we define:
$$a\times b:=\bigwedge_{x\in A}\left((a\rightarrow (b\rightarrow x))\rightarrow x\right)$$
$$a+ b:=\bigwedge_{x\in A}\left((a\rightarrow x)\rightarrow ((b\rightarrow x)\rightarrow x)\right)$$
and for every set indexed family $(a_{i})_{i\in I}$ we define
$$\lar{\forall} _{i\in I} a_i:=\bigwedge_{i\in I}a_{i}\qquad\qquad\lar{\exists} _{i\in I} a_i:=\bigwedge_{x\in A}\left(\bigwedge_{i\in I}\left(a_{i}\rightarrow x\right)\rightarrow x\right)$$
Note that the operations $a\times b$ (implicative conjunction) and $a+b$ (implicative disjunction) are in general not associative, commutative or idempotent on the nose (think of intuitionistic or classical realizability), but they clearly are up to the equivalence $\equiv_{\Sigma}$ (logical equivalence modulo the separator~$\Sigma$).
Also note that unlike universal quantifications (that are simply interpreted as meets), existential quantifications are not interpreted here as joins (as one would expect in forcing or in intuitionistic realizability), but they are rather interpreted using the standard second-order encoding of $\exists$ in minimal second-order logic (thus using the same trick as in classical realizability). The reason is that in the framework of implicative algebras (that contains classical realizability as a particular case), joins do not satisfy the elimination rule of existential quantification, except in particular cases that will be discussed in Section~\ref{ss:CaseIntReal}.
However, the price to pay for this encoding is that the corresponding realizers are in general more complex.

We also introduce shorthands for some $\lambda$-terms: $\overline{\mathbf{k}}:=\lambda x.\lambda y.y$, $\mathbf{p}:=\lambda x.\lambda y.\lambda z.zxy$, $\mathbf{p}_1:=\lambda u.u\mathbf{k}$, $\mathbf{p}_2:=\lambda v.v\overline{\mathbf{k}}$, $\mathbf{j}_1:=\lambda x.\lambda z.\lambda w.zx$, $\mathbf{j}_2:=\lambda x.\lambda z.\lambda w.wx$, $\mathbf{e}:=\lambda x.\lambda z.zx$. Notice that (the encodings of) all of them belong to the separator~$\Sigma$.

If $\Gamma$ is a finite list of variable assignments $x_1:a_1,...,x_n:a_n$ with $a_1,...,a_n\in A$ and $x_1,...,x_n$ distinct variables, and $t$ is a $\lambda$-term with parameters in $A$ and free variables among $x_1,..,x_n$, we write $\Gamma\vdash t:a$ as a shorthand for $t[\Gamma]^{\mathbb{A}}\leq a$ (where $t[\Gamma]$ is the result of the substitution corresponding to $\Gamma$ applied to $t$) and the following rules are sound (this is a little variation on the system of rules presented in \cite{miq1}):
\begin{framed}
$$\cfrac{x:a\in \Gamma}{\Gamma\vdash x:a}\;\qquad \cfrac{}{\Gamma\vdash a:a}\qquad \cfrac{\Gamma\vdash t:a\qquad  a\leq b}{\Gamma\vdash t:b}\qquad \cfrac{\Gamma'\leq \Gamma\qquad \Gamma\vdash t:a}{\Gamma'\vdash t:a}$$
$$ $$
$$\cfrac{\Gamma\vdash t:\bot}{\Gamma\vdash t:a}\qquad \cfrac{\Gamma\vdash t:a}{\Gamma\vdash t:\top}\qquad\cfrac{\Gamma\vdash t:a\rightarrow b\qquad \Gamma\vdash s:a}{\Gamma\vdash ts:b}\qquad \cfrac{\Gamma,x:a\vdash t:b }{\Gamma\vdash \lambda x.t:a\rightarrow b}$$
$$ $$
$$\cfrac{\Gamma\vdash t:a\qquad \Gamma\vdash s:b}{\Gamma\vdash \mathbf{p}ts: a\times b}\qquad \cfrac{\Gamma\vdash t:a\times b}{\Gamma\vdash\mathbf{p}_1t:a}\qquad  \cfrac{\Gamma\vdash t:a\times b}{\Gamma\vdash\mathbf{p}_2t:b}$$
$$ $$
$$\cfrac{\Gamma\vdash t:a}{\Gamma\vdash\mathbf{j}_1t:a+b}\;\;\cfrac{\Gamma\vdash t:b}{\Gamma\vdash\mathbf{j}_2t:a+b}\;\; \cfrac{\Gamma\vdash t:a+b\qquad \Gamma, x:a\vdash u:c\qquad\Gamma, y:b\vdash v:c }{\Gamma\vdash t(\lambda x.u)(\lambda y.v):c}$$
$$ $$
$$\cfrac{\Gamma\vdash t:a_i\,(\textrm{for all }i\in I)}{\Gamma\vdash t: \lar{\forall} _{i\in I}a_i}\qquad \cfrac{\Gamma\vdash t: \lar{\forall} _{i\in I}a_i}{\Gamma\vdash t:a_{\overline{i}}}\;\overline{i}\in I$$
$$ $$
$$\cfrac{\Gamma\vdash t:a_{\overline{i}}}{\Gamma\vdash 	\mathbf{e}t:\lar{\exists} _{i\in I}a_i }\; \overline{i}\in I\qquad \cfrac{\Gamma\vdash t:\lar{\exists} _{i\in I}a_i\qquad \Gamma,x:a_i\vdash u:b\,(\textrm{ for all }i\in I)}{\Gamma\vdash t(\lambda x.u):b}$$
\end{framed}
where $\Gamma'\leq \Gamma$ means that for every declaration $x:a$ in $\Gamma$ we have $x:b$ in $\Gamma'$ for some $b\leq a$.

As shown in \cite{miq1}, to every implicative algebra $\mathbb{A}$ can be associated a tripos (see \cite{HJP} or \cite{VOO08}) 
$$\mathsf{P}_{\mathbb{A}}:\mathbf{Set}^{op}\rightarrow \mathbf{Heyt}$$
by sending every set $I$ to the posetal reflection of the preordered set $(A^I,\vdash_{\Sigma[I]})$ where $\varphi\vdash_{\Sigma[I]}\psi$ if and only if $\bigwedge_{i\in I}(\varphi(i)\rightarrow \psi(i))\in \Sigma$ (we will write $\varphi\equiv_{\Sigma[I]}\psi$ if $\varphi\vdash_{\Sigma[I]}\psi$ and $\psi\vdash_{\Sigma[I]}\varphi$) and every function $f:I\rightarrow J$ to the function induced by the pre-composition function $(-)\circ f:A^J\rightarrow A^I$.
Componentwise use of $\times$, $+$ and $\rightarrow$ defines a Heyting prealgebra structure (which needs not to be complete) on every preorder $(A^I,\vdash_{\Sigma[I]})$, which is preserved by pre-composition.
$\lar{\exists}$ and $\lar{\forall}$ are used to produce left and right adjoints to reindexing maps satisfying Beck-Chevalley condition, while a generic predicate is given by (the equivalence class of) the identity function on $A$. 

A remarkable result in \cite{miq2} is the following:
\begin{theorem}\label{teomiq}
Let $\mathsf{P}:\mathbf{Set}^{op}\rightarrow \mathbf{Heyt}$ be a tripos. Then, there exists an implicative algebra $\mathbb{A}$ such that $\mathsf{P}$ is isomorphic to $\mathsf{P}_{\mathbb{A}}$.
\end{theorem}

Recall also (see e.g.\ \cite{VOO08}) that to every tripos $\mathsf{P}$ over $\mathbf{Set}$ is associated an elementary topos $\mathbf{Set}[\mathsf{P}]$ obtained by means of the so-called ``tripos-to-topos'' construction (see \cite{HJP}) whose internal logic can be reduced to that of the corresponding tripos as shown e.g.\ in \cite{VOO08}.

\section{Implicative models of $\mathbf{(I)ZF}$}\label{s:ImpModels}
To define our implicative models we work in $\mathbf{ZFC}$ as metatheory and for our convenience (see Remark \ref{classmodel}) we further assume a strongly inaccessible cardinal $\kappa$ to exist.
Let now $\mathbb{A}$ be a fixed implicative algebra with $|A|<\kappa$.

We define the following hierarchy of sets indexed by ordinals:
$$W_\alpha^{\mathbb{A}}:=\begin{cases} 
\emptyset\textrm{ if }\alpha=0\\
\mathsf{Part}(W_{\beta}^{\mathbb{A}},A)\textrm{ if }\alpha=\beta+1\\
\bigcup_{\beta<\alpha}W_{\beta}^{\mathbb{A}}\textrm{ if }\alpha\textrm{ is a limit ordinal}\\
\end{cases}$$
where $\mathsf{Part}(X,Y)$ denotes the set of partial functions from $X$ to $Y$.
We take $\mathbf{W}$ to be $W^{\mathbb{A}}_{\kappa}$.
Since $W_{\alpha}^{\mathbb{A}}\subseteq W_{\beta}^{\mathbb{A}}$ if $\alpha<\beta$, one can assign a rank in the hierarchy to every element of $\mathbf{W}$ in the obvious way.
In particular, we can define simultaneously, by recursion on rank, two functions $\in_\mathbf{W}, =_{\mathbf{W}}:\mathbf{W}\times \mathbf{W}\rightarrow A$:
\begin{enumerate}
\item $\alpha\in_{\mathbf{W}}\beta:=\lar{\exists} _{t\in\partial_0(\beta)}\left( \beta(t)\times (t=_{\mathbf{W}}\alpha) \right)$\footnote{We denote with $\partial_0(f)$ the domain of a partial function $f$, that is the set of those $x$ for which $f(x)$ is defined.}
\item $\alpha=_{\mathbf{W}}\beta:=(\alpha\subseteq_{\mathbf{W}}\beta)\times(\beta\subseteq_{\mathbf{W}}\alpha) $
where $\alpha\subseteq_{\mathbf{W}}\beta:=\lar{\forall} _{t\in \partial_0(\alpha)}\left(\alpha(t)\rightarrow t\in_\mathbf{W}\beta\right)$.
\end{enumerate}
We interpret the language of set theory in such a way that to every formula in context $\varphi[x_{1},...,x_n]$ we associate a function $$\left\|\varphi[x_{1},...,x_n]\right\|:\mathbf{W}^{n}\rightarrow A\footnote{If $n=0$, then $\left\|\varphi\,[\,]\right\|$ is identified with an element of $A$.}$$
by recursion on complexity of formulas as follows:
\begin{enumerate}
\item[]
\item $\left\|x_i\in x_j\,[x_1,...,x_n]\right\|(\alpha_1,...,\alpha_n):\equiv \alpha_i\in_{\mathbf{W}}\alpha_j$
\item $\left\|x_i= x_j\,[x_1,...,x_n]\right\|(\alpha_1,...,\alpha_n):\equiv\alpha_i=_{\mathbf{W}}\alpha_j$
\item $\left\|\bot[\underline{x}]\right\|(\underline{\alpha}):\equiv \bot$
\item $\left\|\varphi\wedge \psi[\underline{x}]\right\|(\underline{\alpha}):\equiv \left\|\varphi[\underline{x}]\right\|(\underline{\alpha})\times \left\|\psi[\underline{x}]\right\|(\underline{\alpha})$
\item $\left\|\varphi\vee \psi[\underline{x}]\right\|(\underline{\alpha}):\equiv \left\|\varphi[\underline{x}]\right\|(\underline{\alpha})+ \left\|\psi[\underline{x}]\right\|(\underline{\alpha})$
\item $\left\|\varphi\rightarrow \psi[\underline{x}]\right\|(\underline{\alpha}):\equiv \left\|\varphi[\underline{x}]\right\|(\underline{\alpha})\rightarrow \left\|\psi[\underline{x}]\right\|(\underline{\alpha})$
\item $\left\|\exists y\,\varphi\,[\underline{x}]\right\|(\underline{\alpha}):\equiv \lar{\exists} _{\beta\in \mathbf{W}}\left(\left\|\varphi\,[\underline{x},y]\right\|(\underline{\alpha},\beta)\right)$
\item $\left\|\forall y\,\varphi\,[\underline{x}]\right\|(\underline{\alpha}):\equiv \lar{\forall} _{\beta\in \mathbf{W}}\left(\left\|\varphi\,[\underline{x},y]\right\|(\underline{\alpha},\beta)\right)$\footnote{In clauses (vii) and (viii), we assume, without loss of generality, that $y$ is not a variable in the context $[\underline{x}]$.}
\item[]
\end{enumerate}
We write $\mathbf{W}\vDash \varphi\,[\underline{x}]$ for $\bigwedge_{\underline{\alpha}\in \mathbf{W}^{n}}\left(\left\|\varphi\,[\underline{x}]\right\|(\underline{\alpha})\right)\in \Sigma$ when $[\underline{x}]$ has length $n>0$, and for every closed formula $\varphi$ we write $\mathbf{W}\vDash \varphi$ for $\left\|\varphi\,[\,]\right\|\in \Sigma$.
Thus, $\mathbf{W}\vDash \varphi\,[\underline{x}]$ just means that $\left\|\varphi[\underline{x}]\right\|$ is in the maximal class of $\mathsf{P}_{\mathbb{A}}(\mathbf{W}^{n})$, where $n$ is the length of the context of variables $[\underline{x}]$. 
We will often write $\left\|\varphi\right\|$ instead of $\left\|\varphi[\,]\right\|$.\\

\begin{remark}\label{classmodel}
  Note that here, we chose to construct the model $\mathbf{W}$ as a set, so that we can use it later (together with the suitable $\mathbb{A}$-equivalence) as an object of the topos $\mathbf{Set}[\mathsf{P}_{\mathbb{A}}]$ induced by the implicative algebra~$\mathbb{A}$ (cf Section~\ref{s:ModTopos}).
  However, if one is only interested in the set-theoretic part of the work, it is actually simpler to construct the model $\mathbf{W}$ as a proper class (as it is traditionally done in forcing or in intuitionistic or classical realizability), thus removing the need of assuming the existence of an inaccessible cardinal~$\kappa$ (whose only purpose is to make the model $\mathbf{W}$ fit into a set).
\end{remark}

\subsection{Useful lemmas}

\begin{lemma}\label{not}
There exist $\mathbf{\rho}, \mathbf{j},\mathbf{\sigma},\mathbf{s}_1,\mathbf{s}_2,\mathbf{s}_3\in \Sigma$ such that 
\begin{enumerate}
\item $\mathbf{\rho}\leq \bigwedge_{\alpha\in \mathbf{W}}\left(\alpha=_{\mathbf{W}}\alpha\right)$
\item $\mathbf{j}\leq \bigwedge_{\alpha\in \mathbf{W}}\bigwedge_{u\in \partial_0(\alpha)}\left(\alpha(u)\rightarrow u\in_{\mathbf{W}}\alpha\right)$
\item $\mathbf{\sigma}\leq\bigwedge_{\alpha,\beta\in \mathbf{W}}\left(\alpha=_{\mathbf{W}}\beta\rightarrow \beta=_{\mathbf{W}}\alpha\right)$
\item $\mathbf{s}_1\leq \bigwedge_{\alpha,\beta,\gamma\in \mathbf{W}}\left(\alpha=_{\mathbf{W}}\beta\times \gamma\in_{\mathbf{W}} \alpha\rightarrow \gamma\in_{\mathbf{W}}\beta\right)$
\item $\mathbf{s}_2\leq \bigwedge_{\alpha,\beta,\gamma\in \mathbf{W}}\left(\alpha=_{\mathbf{W}}\beta\times \alpha\in_{\mathbf{W}} \gamma\rightarrow \beta\in_{\mathbf{W}}\gamma\right)$
\item $\mathbf{s}_3\leq \bigwedge_{\alpha,\beta,\gamma\in \mathbf{W}}\left(\alpha=_{\mathbf{W}}\beta\times \gamma=_{\mathbf{W}} \alpha\rightarrow \gamma=_{\mathbf{W}}\beta\right)$
\end{enumerate}
\end{lemma}

\begin{proof}
\begin{enumerate}
\item Let $\mathbf{\rho}$ be $\mathbf{y}f\in \Sigma$ where $f:=\lambda r.\mathbf{p}(\lambda x.\mathbf{e}(\mathbf{p}xr))(\lambda x.\mathbf{e}(\mathbf{p}xr))$ and $\mathbf{y}$ is a pure closed $\lambda$-term which is a fixed point operator such that $\mathbf{y}f$ $\beta$-reduces to $f(\mathbf{y}f)$ for every $f$ (see e.g.\ \cite{VOO08}). We claim that $\mathbf{\rho}\leq \alpha=_{\mathbf{W}}\alpha$ for every $\alpha\in \mathbf{W}$. Let $\alpha$ be an arbitrary element of $\mathbf{W}$ and let us assume that $\mathbf{\rho}\leq \beta=_{\mathbf{W}}\beta$ for every $\beta\in \mathbf{W}$ with rank in the hierarchy strictly less than that of  $\alpha$ (and thus in particular for every element of the domain of $\alpha$). Then we can consider the following derivation tree in which we used only rules from the previous section.
$$\cfrac{\cfrac{\cfrac{\cfrac{\cfrac{\cfrac{}{x:\alpha(u)\vdash x:\alpha(u)\;(\textrm{for all }u\in \partial_{0}(\alpha))}\qquad \cfrac{}{x:\alpha(u)\vdash \mathbf{\rho}:u=_{\mathbf{W}}u\;(\textrm{for all }u\in \partial_{0}(\alpha))}}{x:\alpha(u)\vdash \mathbf{p}x\mathbf{\rho}: \alpha(u)\times u=_{\mathbf{W}}u \;(\textrm{for all }u\in \partial_{0}(\alpha))}}{x:\alpha(u)\vdash \mathbf{e}(\mathbf{p}x\mathbf{\rho}): u\in_{\mathbf{W}}\alpha \;(\textrm{for all }u\in \partial_{0}(\alpha))}}{\vdash\lambda x.\mathbf{e}(\mathbf{p}x\mathbf{\rho}):\alpha(u)\rightarrow u\in_{\mathbf{W}}\alpha \;(\textrm{for all }u\in \partial_{0}(\alpha)) }}{\vdash\lambda x.\mathbf{e}(\mathbf{p}x\mathbf{\rho}):\alpha\subseteq_{\mathbf{W}}\alpha}}{\vdash\mathbf{p}(\lambda x.\mathbf{e}(\mathbf{p}x\mathbf{\rho}))(\lambda x.\mathbf{e}(\mathbf{p}x\mathbf{\rho})):\alpha=_{\mathbf{W}}\alpha}$$

The last $\lambda$-term in the deduction tree is a $\beta$ reduction of $\mathbf{\rho}$. Thus we can conclude that $\mathbf{\rho}\leq \alpha=_{\mathbf{W}}\alpha$.

\item Let $\mathbf{j}$ be defined as $\lambda x.\mathbf{e}(\mathbf{p}x\mathbf{\rho})\in \Sigma$. Assume $\alpha\in \mathbf{W}$ and $u\in \partial_{0}(\alpha)$. Then $x: \alpha(u)\vdash\mathbf{p}x\mathbf{\rho}: \alpha(u)\times u=_{\mathbf{W}}u$ . Hence, $x: \alpha(u)\vdash \mathbf{e}(\mathbf{p}x\mathbf{\rho}): u\in_{\mathbf{W}}\alpha$. Thus $$\mathbf{j}\leq \bigwedge_{\alpha\in \mathbf{W}}\bigwedge_{u\in \partial_0(\alpha)}\left(\alpha(u)\rightarrow u\in_{\mathbf{W}}\alpha\right)$$
\item $\mathbf{\sigma}$ can be just defined as $\lambda x.\mathbf{p}(\mathbf{p}_2x)(\mathbf{p}_1x) \in \Sigma$
\item[(iv),(v),(vi)] Assume $\mathbf{s}_3$ to exist. 

Let $\alpha,\beta,\gamma\in \mathbf{W}$ and let $\Gamma(u)$ be a shorthand for 
$$x:\alpha=_{\mathbf{W}}\beta\times \alpha\in_{\mathbf{W}}\gamma, y:\gamma(u)\times u=_{\mathbf{W}}\alpha$$ where $u$ an arbitrary element of the domain of $\gamma$. 
Since 
$$\vdash \mathbf{s}_3:\alpha=_{\mathbf{W}}\beta\times u=_{\mathbf{W}}\alpha\rightarrow u=_{\mathbf{W}}\beta$$
we obtain $\Gamma(u)\vdash \mathbf{p}(\mathbf{p}_1y)(\mathbf{s}_3(\mathbf{p}(\mathbf{p}_1x)(\mathbf{p}_2y))):\gamma(u)\times u=_{\mathbf{W}}\beta $ from which it follows that 
$$\Gamma(u)\vdash \mathbf{e}(\mathbf{p}(\mathbf{p}_1y)(\mathbf{s}_3(\mathbf{p}(\mathbf{p}_1x)(\mathbf{p}_2y)))):\beta\in_{\mathbf{W}}\gamma$$
Since $x:\alpha=_{\mathbf{W}}\beta\times \alpha\in_{\mathbf{W}}\gamma\vdash \mathbf{p}_2 x:\alpha\in_{\mathbf{W}}\gamma$, we get
 $$x:\alpha=_{\mathbf{W}}\beta\times \alpha\in_{\mathbf{W}}\gamma\vdash (\mathbf{p}_2x)(\lambda y.\mathbf{e}(\mathbf{p}(\mathbf{p}_1y)(\mathbf{s}_3(\mathbf{p}(\mathbf{p}_1x)(\mathbf{p}_2y))))):\beta\in_{\mathbf{W}}\gamma$$
 from which it follows that 
 $$\vdash\lambda x.(\mathbf{p}_2x)(\lambda y.(\mathbf{e}(\mathbf{p}(\mathbf{p}_1y)(\mathbf{s}_3(\mathbf{p}(\mathbf{p}_1x)(\mathbf{p}_2y)))))):\alpha=_{\mathbf{W}}\beta\times \alpha\in_{\mathbf{W}}\gamma\rightarrow \beta\in_{\mathbf{W}} \gamma$$
From this it follows that $\mathbf{s}_2$ can be defined as $\lambda x.(\mathbf{p}_2x)(\lambda y.(\mathbf{e}(\mathbf{p}(\mathbf{p}_1y)(\mathbf{s}_3(\mathbf{p}(\mathbf{p}_1x)(\mathbf{p}_2y))))))$.

Assume now $\mathbf{s}_2$ to exist and consider $\Gamma'(u)$ a shorthand for 
$$x:\alpha=_{\mathbf{W}}\beta\times \gamma\in_{\mathbf{W}}\alpha,y:\alpha(u)\times u=_{\mathbf{W}}\gamma$$ where $u$ is an arbitrary element of the domain of $\alpha$. We easily see that 
$$\Gamma'(u)\vdash \mathbf{p}(\mathbf{p}_2y)((\mathbf{p_1}(\mathbf{p}_1x))(\mathbf{p}_1y)):u=_{\mathbf{W}}\gamma \times u\in_{\mathbf{W}}\beta $$ Thus
$$\Gamma'(u)\vdash\mathbf{s}_2( \mathbf{p}(\mathbf{p}_2y)((\mathbf{p_1}(\mathbf{p}_1x))(\mathbf{p}_1y))):\gamma\in_{\mathbf{W}}\beta $$
Since $x:\alpha=_{\mathbf{W}}\beta\times \gamma\in_{\mathbf{W}}\alpha\vdash \mathbf{p}_2x:\gamma\in_{\mathbf{W}}\alpha$, we have that
$$x:\alpha=_{\mathbf{W}}\beta\times \gamma\in_{\mathbf{W}}\alpha\vdash (\mathbf{p}_2x)(\lambda y.\mathbf{s}_2( \mathbf{p}(\mathbf{p}_2y)((\mathbf{p_1}(\mathbf{p}_1x))(\mathbf{p}_1y)))):\gamma\in_{\mathbf{W}}\beta$$
from which it follows that 
$$\vdash \lambda x.(\mathbf{p}_2x)(\lambda y.\mathbf{s}_2( \mathbf{p}(\mathbf{p}_2y)((\mathbf{p_1}(\mathbf{p}_1x))(\mathbf{p}_1y)))):\alpha=_{\mathbf{W}}\beta\times \gamma\in_{\mathbf{W}}\alpha\rightarrow\gamma\in_{\mathbf{W}}\beta$$
Thus $\mathbf{s}_1$ can be defined as $\lambda x.(\mathbf{p}_2x)(\lambda y.\mathbf{s}_2( \mathbf{p}(\mathbf{p}_2y)((\mathbf{p_1}(\mathbf{p}_1x))(\mathbf{p}_1y))))$.

Similarly, one can prove that if $\mathbf{s}_1$ is assumed to exist, then one can define $\mathbf{s}_3$ as a $\lambda$-term containing $\mathbf{s}_1$ as the unique parameter.

The idea is now to compose this mutual dependences to define, by a fix point $\mathbf{y}g$, one among $\mathbf{s}_1$, $\mathbf{s}_2$ and $\mathbf{s}_3$, and then define the other two using that one. So for example, if we define $\mathbf{s}_3$ as a fixpoint, we can then define $\mathbf{s}_2$ using $\mathbf{s}_3$ and then $\mathbf{s}_1$ using $\mathbf{s}_2$.
This works well since composing the proofs above one can see that $\mathbf{s}_3\leq \bigwedge_{\alpha,\beta,\gamma\in \mathbf{W}}\left(\alpha=_{\mathbf{W}}\beta\times \gamma=_{\mathbf{W}} \alpha\rightarrow \gamma=_{\mathbf{W}}\beta\right)$ whenever $\mathbf{s}_3\leq u=_{\mathbf{W}}v\times w=_{\mathbf{W}} u \rightarrow w=_{\mathbf{W}}v$ for every $u,v,w$ with rank strictly less than the maximum of the ranks of $\alpha$, $\beta$ and $\gamma$.
\end{enumerate}
\end{proof}

Given two lists of parameters $\underline{\alpha}=\alpha_1,\ldots,\alpha_n\in\mathbf{W}^n$ and $\underline{\beta}=\beta_1,\ldots,\beta_n\in\mathbf{W}^n$ (for $n\ge 0$), we write
$$\underline{\alpha}=_{\mathbf{W}}\underline{\beta}~:=~
(\cdots((\alpha_1=_{\mathbf{W}}\beta_1\times
\alpha_2=_{\mathbf{W}}\beta_2)\times
\alpha_3=_{\mathbf{W}}\beta_3)\cdots)\times
\alpha_n=_{\mathbf{W}}\beta_n\eqno(\in A)$$
simply letting $\underline{\alpha}=_{\mathbf{W}}\underline{\beta}:=\top$
in the particular case where $n=0$.
Note that the element $\underline{\alpha}=_{\mathbf{W}}\underline{\beta}$ ($\in A$) depends on the order of the parameters $\underline{\alpha}=\alpha_1,\ldots,\alpha_n$ and $\underline{\beta}=\beta_1,\ldots,\beta_n$ (and on the choice to associate $\times$'s to the left) when considered on the nose, but up to the equivalence $\equiv_{\Sigma}$, it is of course invariant under any (common) permutation of the parameters $\underline{\alpha}$ and $\underline{\beta}$.

\begin{lemma}\label{subfor}For every formula in context $\varphi\,[\underline{x}]$ where $\underline{x}$ has length $n>0$, there exists $\mathbf{r}^{\varphi[\underline{x}]}\in \Sigma$ such that 
$$\mathbf{r}^{\varphi[\underline{x}]}\leq \bigwedge_{\underline{\alpha}\in \mathbf{W}^{n}}\bigwedge_{\underline{\beta}\in \mathbf{W}^{n}}\left( \underline{\alpha}=_{\mathbf{W}}\underline{\beta}\times \left\|\varphi\,[\underline{x}]\right\|(\underline{\alpha})\rightarrow \left\|\varphi\,[\underline{x}]\right\|(\underline{\beta})\right)\,.$$
\end{lemma}

\begin{proof}
By induction on complexity of formulas by using the previous lemma for the atomic cases.
\end{proof}
Also the following lemma can be easily proved as a consequence of the previous results and of the rules in the previous section.
\begin{lemma}\label{intlog} Let $\varphi[\underline{x}]$ and $\psi[\underline{x}]$ be formulas in context in the language of set theory and let $n$ be the length of $[\underline{x}]$.  
If $\varphi\vdash^{\underline{x}}_{\mathbf{IL}^=}\psi$, then $\left\|\varphi\,[\underline{x}]\right\|\vdash_{\Sigma[\mathbf{W}^{n}]}\left\|\psi\,[\underline{x}]\right\|$ (where with $\mathbf{IL}^=$ we denote first-order intuitionistic logic with equality on the language of $\mathbf{(I)ZF}$). In the case in which $\mathbb{A}$ is a classical implicative algebra, if $\varphi\vdash^{\underline{x}}_{\mathbf{CL}^=}\psi$, then $\left\|\varphi\,[\underline{x}]\right\|\vdash_{\Sigma[\mathbf{W}^{n}]}\left\|\psi\,[\underline{x}]\right\|$  (where with $\mathbf{CL}^=$ we denote first-order classical logic with equality on the language of $\mathbf{(I)ZF}$).
\end{lemma}

\begin{lemma}\label{rest}If $[\underline{x}]$ has length $n$, then
$$\left\|\exists z\in y\,\varphi\,[\underline{x},y]\right\|\equiv_{\Sigma[\mathbf{W}^{n+1}]}\Lambda \underline{\alpha}.\Lambda \beta.\lar{\exists} _{u\in \partial_{0}(\beta)}\left(\beta(u)\times \left\|\varphi\,[\underline{x},y,z]\right\|(\underline{\alpha},\beta,u) \right)\footnote{We use the notation $\Lambda \alpha.f(\alpha)$ to denote the function sending each $\alpha$ in the domain to $f(\alpha)$.}$$
$$\left\|\forall z\in y\,\varphi\,[\underline{x},y]\right\|\equiv_{\Sigma[\mathbf{W}^{n+1}]}\Lambda \underline{\alpha}.\Lambda \beta.\lar{\forall} _{u\in \partial_{0}(\beta)}\left(\beta(u)\rightarrow \left\|\varphi\,[\underline{x},y,z]\right\|(\underline{\alpha},\beta,u) \right)$$
\end{lemma}
\begin{proof}
We consider the case of the existential quantifier and we leave the analogous proof of the universal case to the reader. We also restrict to the case in which $\underline{x}$ is empty. The general case is analogous, but only heavier in notation.
By definition of the interpretation we have that 
$\left\|\exists z\in y\,\varphi\,[y]\right\|(\beta)$ is $$\eta:=\lar{\exists} _{\gamma \in \mathbf{W}}\left(\lar{\exists} _{u\in \partial_{0}(\beta)}(\beta(u)\times u=_{\mathbf{W}}\gamma)\times \left\|\varphi\,[y,z]\right\|(\beta,\gamma)\right)$$
We denote with $\eta_{1}(\gamma)$ the scope of the quantifier $\lar{\exists} _{\gamma \in \mathbf{W}}$, while we denote with $\eta_2(u,\gamma)$ the scope of the quantifier $\lar{\exists} _{u\in \partial_{0}(\beta)}$.
It is immediate to check that the following sequent holds for every $\beta,\gamma\in \mathbf{W}$ and every $u\in \partial_0(\beta)$: 
$$x:\eta, y:\eta_1(\gamma),z:\eta_2(u,\gamma)\vdash \mathbf{l}:(\beta=_{\mathbf{W}}\beta\times \gamma=_{\mathbf{W}}u)\times \left\|\varphi\,[y,z]\right\|(\beta,\gamma) )$$
where $\mathbf{l}:=\mathbf{p}(\mathbf{p}\rho(\sigma(\mathbf{p}_2z)))(\mathbf{p}_2y)$.

With the notation from Lemma \ref{subfor}, we can conclude that 
$$x:\eta, y:\eta_1(\gamma),z:\eta_2(u,\gamma)\vdash \mathbf{r}^{\varphi[y,z]}\mathbf{l}:\left\|\varphi\,[y,z]\right\|(\beta,u) $$ Thus $x:\eta, y:\eta_1(\gamma),z:\eta_2(u,\gamma)\vdash \mathbf{p}(\mathbf{p}_1z)(\mathbf{r}^{\varphi[y,z]}\mathbf{l}):\beta(u)\times\left\|\varphi\,[y,z]\right\|(\beta,u) $.
Hence
$$x:\eta, y:\eta_1(\gamma),z:\eta_2(u,\gamma)\vdash \mathbf{e}(\mathbf{p}(\mathbf{p}_1z)(\mathbf{r}^{\varphi[y,z]}\mathbf{l})):\exists_{w\in \partial_{0}(\beta)}(\beta(w)\times\left\|\varphi\,[y,z]\right\|(\beta,w))$$
Using the rules of elimination of existential quantification, one can conclude that 
$$\eta\rightarrow \exists_{w\in \partial_{0}(\beta)}(\beta(w)\times\left\|\varphi\,[y,z]\right\|(\beta,w))\geq \mathbf{l}'\in \Sigma$$
where $\mathbf{l}':=\lambda x.x\lambda y.(\mathbf{p}_1y)(\lambda z.\mathbf{e}(\mathbf{p}(\mathbf{p}_1z)(\mathbf{r}^{\varphi[y,z]}\mathbf{l})))$.

It is easier to show that for every $\beta\in \mathbf{W}$
$$ \exists_{w\in \partial_{0}(\beta)}(\beta(w)\times\left\|\varphi\,[y,z]\right\|(\beta,w))\rightarrow \eta \geq \lambda x.x\lambda y.\mathbf{e}(\mathbf{p}(\mathbf{e}(\mathbf{p}(\mathbf{p}_1y)\rho))(\mathbf{p}_2y))\in \Sigma$$
 One can in fact write a deduction tree in which the existential quantifiers of the consequent are both witnessed by a $w\in \partial_{0}(\beta)$ for which $\beta(w)\times\left\|\varphi\,[y,z]\right\|(\beta,w)$ is assumed to hold.

\end{proof}
\begin{corollary}\label{sub} $\left\|x\subseteq y\,[x,y]\right\|\equiv_{\Sigma[\mathbf{W}^2]}\Lambda \alpha.\Lambda \beta.\left(\alpha\subseteq_{\mathbf{W}}\beta\right)$.
\end{corollary}

\subsection{Validity of axioms}
In this subsection we show that the interpretation we gave is in fact a model of $\mathbf{IZF}$ (when $\mathbb{A}$ is not classical) or a model of $\mathbf{ZF}$ (when $\mathbb{A}$ is classical), that is, if $\mathbf{(I)ZF}\vdash \varphi$, then $\mathbf{W}\vDash \varphi$. In order to show this we prove that every axiom of $\mathbf{(I)ZF}$ is valid in the interpretation.
\subsubsection{Extensionality}
Thanks to Corollary \ref{sub} we know that
$$\left\|\mathbf{Ext}\right\|\equiv_{\Sigma}\lar{\forall} _{\alpha\in \mathbf{W}}\lar{\forall} _{\beta\in \mathbf{W}}\left(\alpha\subseteq_{\mathbf{W}}\beta \times \beta\subseteq_{\mathbf{W}}\alpha\rightarrow \alpha=_{\mathbf{W}}\beta\right)\geq \lambda x.x \in \Sigma$$
Thus, $\mathbf{W}\vDash \mathbf{Ext}$.
\subsubsection{Pair}
$$\left\|\mathbf{Pair}\right\|\equiv\lar{\forall} _{\alpha\in \mathbf{W}}\lar{\forall} _{\beta\in \mathbf{W}}\lar{\exists} _{\gamma\in \mathbf{W}}\left(\alpha\in_{\mathbf{W}}\gamma \times \beta\in_{\mathbf{W}}\gamma\right)$$
Let us consider arbitrary $\alpha,\beta\in \mathbf{W}$ and the partial function $\eta_{\alpha,\beta}\in \mathbf{W}$ defined as follows: $\partial_{0}(\eta_{\alpha,\beta})=\{\alpha, \beta\}$ and $\eta_{\alpha,\beta}(u)=\top$ for every $u$ in the domain. By Lemma \ref{not}, $\vdash\mathbf{\rho}: \alpha=_{\mathbf{W}}\alpha$, from which it follows that 
$$\vdash\mathbf{q}:=\mathbf{e}(\mathbf{p}\top\rho): \exists_{t\in \{\alpha,\beta\}}\left(\top\times t=_{\mathbf{W}}\alpha\right)=\alpha\in_{\mathbf{W}}\eta_{\alpha,\beta}$$ In the same way, we can show that $\vdash\mathbf{q}: \beta\in_{\mathbf{W}}\eta_{\alpha,\beta}$. As a consequence 
$$\vdash\mathbf{q}':=\mathbf{p}\mathbf{q}\mathbf{q}: \alpha\in_{\mathbf{W}}\eta_{\alpha,\beta}\times\beta\in_{\mathbf{W}}\eta_{\alpha,\beta}$$ Thus,  $\vdash\mathbf{e}\mathbf{q}': \lar{\exists} _{\gamma\in \mathbf{W}}\left(\alpha\in_{\mathbf{W}}\gamma \times \beta\in_{\mathbf{W}}\gamma\right)$.  Since $\mathbf{e}\mathbf{q}'$ does not depend on $\alpha$ and $\beta$ and belongs to $\Sigma$,  
$$\vdash\mathbf{e}\mathbf{q}': \lar{\forall} _{\alpha\in \mathbf{W}}\lar{\forall} _{\beta\in \mathbf{W}}\lar{\exists} _{\gamma\in \mathbf{W}}\left(\alpha\in_{\mathbf{W}}\gamma \times \beta\in_{\mathbf{W}}\gamma\right)$$
Hence $\mathbf{W}\vDash \mathbf{Pair}$.

\subsubsection{Union}
Thanks to Lemma \ref{rest} we know that
$$\left\|\mathbf{Union}\right\|\equiv_{\Sigma}\lar{\forall} _{\alpha\in \mathbf{W}}\lar{\exists} _{\beta\in \mathbf{W}}\lar{\forall} _{u\in \partial_{0}(\alpha)}\left(\alpha(u)\rightarrow\lar{\forall} _{w\in \partial_{0}(u)}\left(u(w)\rightarrow w\in_{\mathbf{W}}\beta\right)\right)$$
Let us fix now an arbitrary $\alpha\in \mathbf{W}$ and let us define $\zeta_{\alpha}\in \mathbf{W}$ as follows. The domain of $\zeta_{\alpha}$ is $\bigcup_{u\in \partial_{0}(\alpha)}\partial_{0}(u)$ and for every element of such domain $\zeta_{\alpha}(u)=\top$.

Let $w\in \bigcup_{u\in \partial_{0}(\alpha)}\partial_{0}(u)$, then $\vdash\mathbf{p}\top \mathbf{\rho}: \top\times w=_{\mathbf{W}}w$ (we are using Lemma \ref{not}), from which it follows that $\vdash\mathbf{e}(\mathbf{p}\top \mathbf{r})): w\in_{\mathbf{W}}\zeta_{\alpha}$. From this it follows that 
$$\vdash\lambda v'.\lambda v.\mathbf{e}(\mathbf{p}\top \mathbf{r})): \lar{\forall} _{u\in \partial_{0}(\alpha)}\left(\alpha(u)\rightarrow\lar{\forall} _{w\in \partial_{0}(u)}\left(u(w)\rightarrow w\in_{\mathbf{W}}\zeta_{\alpha}\right)\right)$$
and thus that 
$$\vdash\mathbf{e}(\lambda v'.\lambda v.\mathbf{e}(\mathbf{p}\top \mathbf{r}))):  \lar{\exists} _{\beta\in \mathbf{W}}\lar{\forall} _{u\in \partial_{0}(\alpha)}\left(\alpha(u)\rightarrow\lar{\forall} _{w\in \partial_{0}(u)}\left(u(w)\rightarrow w\in_{\mathbf{W}}\beta\right)\right)$$
Since $\mathbf{e}(\lambda v'.\lambda v.\mathbf{e}(\mathbf{p}\top \mathbf{r})))$ does not depend on $\alpha$ and it is an element of $\Sigma$, we have 
$$\vdash\mathbf{e}(\lambda v'.\lambda v.\mathbf{e}(\mathbf{p}\top \mathbf{r}))): \lar{\forall} _{\alpha\in \mathbf{W}}\lar{\exists} _{\beta\in \mathbf{W}}\lar{\forall} _{u\in \partial_{0}(\alpha)}\left(\alpha(u)\rightarrow\lar{\forall} _{w\in \partial_{0}(u)}\left(u(w)\rightarrow w\in_{\mathbf{W}}\beta\right)\right)$$
and $\mathbf{W}\vDash \mathbf{Union}$.
\subsubsection{Powerset}
Using Corollary \ref{sub}, $\left\|\mathbf{Pow}\right\|\equiv_{\Sigma}\lar{\forall} _{\alpha\in \mathbf{W}}\lar{\exists} _{\beta\in \mathbf{W}}\lar{\forall} _{\gamma\in \mathbf{W}}\left(\gamma\subseteq_{\mathbf{W}} \alpha\rightarrow \gamma\in_{\mathbf{W}} \beta\right)$.

Let us consider an arbitrary $\alpha\in \mathbf{W}$ and define $\pi_{\alpha}\in \mathbf{W}$ as that partial function having domain $A^{\partial_0(\alpha)}$ and for which $\pi_{\alpha}(u)=\top$ for every $u$ in the domain.  

For every $\gamma\in \mathbf{W}$ we also define $\gamma_{\alpha}\in \mathbf{W}$ as follows. The domain of $\gamma_{\alpha}$ is $\partial_{0}(\alpha)$ and $\gamma_{\alpha}(u):=u\in_{\mathbf{W}}\alpha \times u\in_{\mathbf{W}}\gamma$ for every $u$ in the domain. 

We now use Lemma \ref{not} and its notation.
Let $u\in \partial_{0}(\gamma)$ and $t\in \partial_0(\alpha)$. Then:
\begin{enumerate}
\item $x: \gamma\subseteq_{\mathbf{W}}\alpha, y: \gamma(u), z:\alpha(t)\times t=_{\mathbf{W}}u  \vdash \mathbf{j}(\mathbf{p}_1z):t\in_{\mathbf{W}}\alpha $
\item $x: \gamma\subseteq_{\mathbf{W}}\alpha, y: \gamma(u), z:\alpha(t)\times t=_{\mathbf{W}}u  \vdash \mathbf{p}_2z:t=_{\mathbf{W}}u$
\item $x: \gamma\subseteq_{\mathbf{W}}\alpha, y: \gamma(u), z:\alpha(t)\times t=_{\mathbf{W}}u  \vdash \mathbf{s}_2(\mathbf{p}(\sigma(\mathbf{p}_2z))(\mathbf{j}y)):t\in_{\mathbf{W}}\gamma$
\end{enumerate}
 From this it follows that 
$$x: \gamma\subseteq_{\mathbf{W}}\alpha, y: \gamma(u), z:\alpha(t)\times t=_{\mathbf{W}}u  \vdash \widetilde{\mathbf{r}}:=\mathbf{e}(\mathbf{p}(\mathbf{p}(\mathbf{j}(\mathbf{p}_1z))(\mathbf{s}_2(\mathbf{p}(\sigma(\mathbf{p}_2z))(\mathbf{j}y))))(\mathbf{p}_2z)):u\in_{\mathbf{W}}\gamma_\alpha $$
Since $$x: \gamma\subseteq_{\mathbf{W}}\alpha, y: \gamma(u)\vdash xy: u\in_{\mathbf{W}}\alpha$$
then 
$$x: \gamma\subseteq_{\mathbf{W}}\alpha, y: \gamma(u)\vdash xy(\lambda z.\widetilde{\mathbf{r}}):u\in_{\mathbf{W}}\gamma_\alpha$$
 From this it follows that 
 $$x: \gamma\subseteq_{\mathbf{W}}\alpha\vdash \lambda y.(xy(\lambda z.\widetilde{\mathbf{r}})):\gamma\subseteq_{\mathbf{W}}\gamma_\alpha$$ 
One can also easily show that $\vdash \lambda z.(\mathbf{p}_2z): \gamma_{\alpha}\subseteq_{\mathbf{W}} \gamma$.
Thus 
$$x: \gamma\subseteq_{\mathbf{W}}\alpha\vdash\overline{\mathbf{r}}:=\mathbf{p}\top(\mathbf{p}(\lambda z.(\mathbf{p}_2z))(\lambda y.(xy(\lambda z.\widetilde{\mathbf{r}})))): \top \times \gamma_{\alpha}=_{\mathbf{W}}\gamma$$
Since $\gamma_{\alpha}$ is in the domain of $\pi_\alpha$ we hence have that 
$$x: \gamma\subseteq_{\mathbf{W}}\alpha\vdash\mathbf{e}\overline{\mathbf{r}}: \gamma \in \pi_{\alpha}$$
We can thus conclude that $\vdash\lambda x.\mathbf{e}\overline{\mathbf{r}}:\gamma\subseteq_{\mathbf{W}}\alpha\rightarrow \gamma\in_{\mathbf{W}}\pi_{\alpha}$.
Since $\lambda x.\mathbf{e}\overline{\mathbf{r}}$ and $\mathbf{e}(\lambda x.\mathbf{e}\overline{\mathbf{r}})$ do not depend on $\gamma$ and $\alpha$ we get, 
$$\vdash\mathbf{e}(\lambda x.\mathbf{e}\overline{\mathbf{r}}): \lar{\forall} _{\alpha\in \mathbf{W}}\lar{\exists} _{\beta\in \mathbf{W}}\lar{\forall} _{\gamma\in \mathbf{W}}\left(\gamma\subseteq_{\mathbf{W}} \alpha\rightarrow \gamma\in_{\mathbf{W}} \beta\right)$$
Since $\mathbf{e}(\lambda x.\mathbf{e}\overline{\mathbf{r}})\in \Sigma$, we can conclude that $\mathbf{W}\vDash \mathbf{Pow}$.

\subsubsection{Infinity}
For every $n\in \omega$, we define $\widehat{n}\in \mathbf{W}$ as follows: $\partial_0(\widehat{n})=\{\widehat{m}|\,m<n\}$ and $\widehat{n}(\widehat{m}):=\overline{m}$ where $\overline{m}\in \Sigma$ is Church's encoding of the natural number $m$\footnote{$\overline{0}:=\lambda x.\lambda y.x$ and $\overline{n+1}:=\overline{s}\,\overline{n}$ where $\overline{s}:\lambda z.\lambda x.\lambda y.y(zxy)$}.
We define $\widehat{\omega}$ as the element of $\mathbf{W}$ with domain $\{\widehat{n}|\,n\in \omega\}$ and defined by $\widehat{\omega}(\widehat{n}):=\overline{n}$.

First, if we consider $\widehat{0}=\emptyset$ and we use Lemma \ref{rest}, we can easily see that 
$$\vdash\mathbf{e}(\mathbf{p}\overline{0}\top):\lar{\exists}_{n\in \omega}(\widehat{\omega}(\widehat{n})\times \lar{\forall}_{m<n}(\widehat{n}(\widehat{m})\rightarrow \bot))\equiv_{\Sigma}\left\|\mathbf{Inf}_1(u)[u]\right\|(\widehat{\omega})$$

Moreover, one can construct a closed $\lambda$-term $f$ whose interpretation is in $\Sigma$ such that for every $n,m\in \omega$
$$\begin{cases}f\overline{n}\,\overline{m}\twoheadrightarrow_{\beta}\mathbf{j}_1(\mathbf{e}(\mathbf{p}\overline{m}\rho))\textrm{ if }\overline{m}\neq \overline{n}\\ f\overline{n}\,\overline{m}\twoheadrightarrow_{\beta}\mathbf{j}_2\rho \textrm{ if }\overline{m}=\overline{n}\end{cases}$$

Then, for every $n\in \omega$
$$\vdash \lambda u.f\overline{n}u:\lar{\forall} _{i<n}(\widehat{n+1}(\widehat{i})\rightarrow ((\widehat{i}\in_{\mathbf{W}}\widehat{n})+(\widehat{i}=_{\mathbf{W}}\widehat{n})))$$
Moreover
$$\vdash \lambda x.\mathbf{e}(\mathbf{p}x\rho):\widehat{n}\subseteq_{\mathbf{W}}\widehat{n+1}$$
$$\vdash \mathbf{e}(\mathbf{p}\overline{n}\rho):\widehat{n}\in_{\mathbf{W}}\widehat{n+1}$$
Using these facts, Lemma \ref{rest} and Corollary \ref{sub}, one can easily show that $\left\|\mathbf{Inf}_{2}(u)[u]\right\|(\widehat{\omega})\in \Sigma$.

Thus we can conclude that $\mathbf{W}\vDash \mathbf{Inf}$.

\subsubsection{Separation}
Assume $\varphi\,[\underline{w},x,z]$ be a formula in context with $\underline{w}$ a list of variable of length $n$.
$$\left\|\mathbf{Sep}_{\varphi}\right\|\equiv_{\Sigma}\lar{\forall} _{\underline{\omega}\in \mathbf{W}^{n}}\lar{\forall} _{\alpha\in \mathbf{W}}\lar{\exists} _{\beta\in \mathbf{W}}\Big(
\lar{\forall} _{u\in \partial_{0}(\beta)}(\beta(u)\rightarrow u\in_{\mathbf{W}}\alpha\times \left\|\varphi\,[\underline{w},x,z]\right\|(\underline{\omega},\alpha,u))\times$$ 
$$\lar{\forall} _{u'\in \partial_{0}(\alpha)}(\alpha(u')\rightarrow ( \left\|\varphi\,[\underline{w},x,z]\right\|(\underline{\omega},\alpha,u')\rightarrow u'\in_{\mathbf{W}}\beta))\Big)$$
For an arbitrary $\alpha\in \mathbf{W}$ and $\underline{\omega}\in \mathbf{W}^n$ we define $\alpha_{\varphi}^{\underline{\omega}}\in \mathbf{W}$ as follows: its domain is equal to the domain of $\alpha$, while $\alpha_{\varphi}^{\underline{\omega}}(u):=\alpha(u)\times \left\|\varphi\,[\underline{w},x,z]\right\|(\underline{\omega},\alpha,u)$.
In order to show that $\mathbf{W}\vDash \mathbf{Sep}_{\varphi}$, it is sufficient to find a $t\in \Sigma$ not depending on $\overline{\omega}$ and $\alpha$ such that 
$$\vdash t: \lar{\forall} _{u\in \partial_{0}(\alpha)}(\alpha_{\varphi}^{\overline{\omega}}(u)\rightarrow u\in_{\mathbf{W}}\alpha\times \left\|\varphi\,[\underline{w},x,z]\right\|(\underline{\omega},\alpha,u))\times$$ 
$$\lar{\forall} _{u'\in \partial_{0}(\alpha)}(\alpha(u')\rightarrow ( \left\|\varphi\,[\underline{w},x,z]\right\|(\underline{\omega},\alpha,u')\rightarrow u'\in_{\mathbf{W}}\alpha_{\varphi}^{\underline{\omega}}))$$
But this is immediate to prove, since using Lemma \ref{not}
$$\vdash\lambda x.\mathbf{p}(\mathbf{j}(\mathbf{p}_1x))(\mathbf{p}_2x): \lar{\forall} _{u\in \partial_{0}(\alpha)}(\alpha_{\varphi}^{\overline{\omega}}(u)\rightarrow u\in_{\mathbf{W}}\alpha\times \left\|\varphi\,[\underline{w},x,z]\right\|(\underline{\omega},\alpha,u))$$
$$\vdash\lambda x.\lambda y.\mathbf{e}(\mathbf{p}(\mathbf{p}xy)\mathbf{\rho}): \lar{\forall} _{u'\in \partial_{0}(\alpha)}(\alpha(u')\rightarrow ( \left\|\varphi\,[\underline{w},x,z]\right\|(\underline{\omega},\alpha,u')\rightarrow u'\in_{\mathbf{W}}\alpha_{\varphi}^{\overline{\omega}}))$$

\subsubsection{$\in$-Induction}
We now consider the axiom schema of $\in$-induction and we restrict to the case of a formula in context $\varphi[x]$ since the general case is analogous, but just heavier in notation.
Let $\mathbf{y}$ be the fix-point operator we have already used in the proof of  Lemma \ref{not} such that $\mathbf{y}f$ $\beta$-reduces to $f(\mathbf{y}f)$ for every $f$ and consider 
$$\mathbf{h}:=\mathbf{y}(\lambda h.\lambda x.x(\lambda y.hx))\in \Sigma$$
in such a way that $\mathbf{h}\leq (\lambda h.\lambda x.x(\lambda y.hx))\mathbf{h}\leq \lambda x.x(\lambda y.\mathbf{h}x)$.

Fix an arbitrary $\overline{\alpha}$ and assume that 
$$\mathbf{h}\leq \lar{\forall} _{\alpha\in \mathbf{W}}\left(\lar{\forall} _{u\in \partial_{0}(\alpha)}(\alpha(u)\rightarrow \left\|\varphi[x]\right\|(u))\rightarrow \left\|\varphi[x]\right\|(\alpha)\right)\rightarrow \left\|\varphi[x]\right\|(\beta)$$ for every $\beta$ with rank strictly less than that of $\overline{\alpha}$.
Let us use $\varepsilon^{\alpha}$ as a shorthand for 
$$\lar{\forall} _{u\in \partial_{0}(\alpha)}(\alpha(u)\rightarrow \left\|\varphi[x]\right\|(u))\rightarrow \left\|\varphi[x]\right\|(\alpha)$$ 
and $\varepsilon$ as a shorthand for $\lar{\forall} _{\alpha\in \mathbf{W}}\varepsilon^{\alpha}$.

If we consider the following derivation tree

$$\cfrac{\cfrac{\cfrac{x:\varepsilon\vdash x:\varepsilon}{x:\varepsilon\vdash x:\varepsilon^{\overline{\alpha}}}\qquad \cfrac{\cfrac{\cfrac{\cfrac{\cfrac{\vdash \mathbf{h}:\varepsilon\rightarrow  \left\|\varphi\,[x]\right\|(u)\,(\textrm{ for every }u\in \partial_{0}(\overline{\alpha}))}{x: \varepsilon\vdash \mathbf{h}:\varepsilon\rightarrow  \left\|\varphi\,[x]\right\|(u)(\textrm{ for every }u\in \partial_{0}(\overline{\alpha}))}\qquad x:\varepsilon\vdash x:\varepsilon}{x:\varepsilon\vdash \mathbf{h}x: \left\|\varphi\,[x]\right\|(u)(\textrm{ for every }u\in \partial_{0}(\overline{\alpha}))}}{x:\varepsilon,y: \overline{\alpha}(u)\vdash \mathbf{h}x: \left\|\varphi\,[x]\right\|(u)(\textrm{ for every }u\in \partial_{0}(\overline{\alpha}))}}{x:\varepsilon\vdash \lambda y.\mathbf{h}x: \overline{\alpha}(u)\rightarrow \left\|\varphi\,[x]\right\|(u)(\textrm{ for every }u\in \partial_{0}(\overline{\alpha}))}}{x:\varepsilon\vdash \lambda y.\mathbf{h}x: \lar{\forall} _{u\in \partial_{0}(\overline{\alpha})}(\overline{\alpha}(u)\rightarrow \left\|\varphi\,[x]\right\|(u))}}{x:\varepsilon\vdash x(\lambda y.\mathbf{h}x):\left\|\varphi[x]\right\|(\overline{\alpha})}}{\vdash \lambda x.x(\lambda y.\mathbf{h}x):\varepsilon\rightarrow \left\|\varphi[x]\right\|(\overline{\alpha})}$$

we can conclude that $\mathbf{h}\leq \varepsilon\rightarrow \left\|\varphi[x]\right\|(\overline{\alpha})$.

By transfinite induction we can hence conclude that 
$$\mathbf{h}\leq \lar{\forall}_{\alpha\in \mathbf{W}}(\varepsilon\rightarrow \left\|\varphi[x]\right\|(\alpha))$$
Since $\mathbf{h}\in \Sigma$ and, by using lemmas \ref{intlog} and \ref{rest}, $\lar{\forall}_{\alpha\in \mathbf{W}}(\varepsilon\rightarrow \left\|\varphi[x]\right\|(\alpha))\equiv_{\Sigma}\left\|\in\textrm{-}\mathbf{Ind}_{\varphi}\right\|$, we can conclude that $\mathbf{W}\vDash \in\textrm{-}\mathbf{Ind}_{\varphi}$.

\subsubsection{Collection}
In order to lighten the notation we will consider $\mathbf{Col}_{\varphi}$ for a formula $\varphi$ in context $[x,y]$ (so without any additional parameter). Moreover we will write $\varphi(a,b)$ instad of $\left\|\varphi\,[x,y]\right\|(a,b)$.

Assume $\alpha\in \mathbf{W}$ and $u\in\partial_0(\alpha)$. Since $\kappa$ is inaccessible, $|A|<\kappa$ and $\{\varphi(u,\gamma)|\,\gamma\in \mathbf{W}\}\subseteq A$, there exists $\eta<\kappa$ such that $\lar{\exists} _{\gamma\in \mathbf{W}}(\top\times \varphi(u,\gamma))=\lar{\exists} _{\gamma\in W^{\mathbb{A}}_{\eta}}(\top \times \varphi(u,\gamma))$. We define $\eta_{u}$ to be the minimum such an $\eta$ and we define $\overline{\eta}_{\alpha}:=\bigvee\{\eta_u|\,u\in \partial_{0}(\alpha)\}$ which is strictly less than $\kappa$, since the cardinality of $\partial_0(\alpha)$ is strictly less than $\kappa$ (because $\kappa$ is strongly inaccessible). 
We define $\beta_{\alpha}\in \mathbf{W}$ as the constant function with value $\top$ and domain $W^{\mathbb{A}}_{\overline{\eta}_{\alpha}}$.  Using the calculus we can show that there is an element $r\in \Sigma$ not depending on $\alpha$ such that 
$$\vdash r:\lar{\forall} _{u\in \partial_{0}(\alpha)}\left(\alpha(u)\rightarrow \lar{\exists} _{\gamma\in \mathbf{W}}\varphi(u,\gamma)\right)\rightarrow$$
$$\qquad\qquad\qquad\qquad\qquad\qquad\qquad \lar{\forall} _{u\in \partial_{0}(\alpha)}\left(\alpha(u)\rightarrow \lar{\exists} _{w\in\partial_{0}(\beta_{\alpha})}(\beta_{\alpha}(w)\times\varphi(u,w))\right)$$
and using this fact one can easily show that $\mathbf{Col}_{\varphi}$ is validated in the model.

\section{Relationship with forcing and realizability models of set theory}
\label{s:relationOther}

In this section, we show that the implicative models of $\mathbf{(I)ZF}$ constructed in the previous section encompass Heyting/Boolean-valued models for $\mathbf{(I)ZF}$ \cite{bell,BELL2} and, up to logical equivalence, Friedman/Rosolini/McCarty realizability models for $\mathbf{IZF}$ \cite{Friedman,Rosolini,McCarty} as well as Krivine's realizability models of $\mathbf{ZF}$~\cite{KRIZF1,KRIZF2}.

\subsection{The case of forcing}
When the parameterizing implicative algebra $\mathbb{A}$ of our model is a complete Heyting/Boolean algebra (with a separator reduced to $\{\top\}$), existential quantifications $\lar{\exists}_{i\in I}a_i$ coincide with suprema $\bigvee_{i\in I}a_i$ whereas implicative conjunctions $a\times b$ coincide with binary meets $a\wedge b$, as shown in~\cite{miq1}.
So that in this case, our implicative model of set theory boils down to the Heyting/Boolean-valued model of $\mathbf{(I)ZF}$ induced by $\mathbb{A}$, such as described e.g.\ in \cite{bell,BELL2}.
Therefore forcing models of set theory (both in intuitionistic and classical logic) appear to be instances of our construction.

\subsection{The case of intuitionistic realizability}\label{ss:CaseIntReal}
The case of intuitionistic realizability corresponds to the implicative algebras~$\mathbb{A}$ that are \emph{compatible with joins}, namely: the implicative algebras satisfying the additional requirement that
$$\bigwedge_{i\in I}(a_i\rightarrow b)
=\Bigl(\bigvee_{i\in I}a_i\Bigr)\rightarrow b$$
for every family $(a_{i})_{i\in I}$ of elements of $\mathbb{A}$ and for every~$b$ in $\mathbb{A}$.

Typical examples of implicative algebras that are compatible with joins are the ones coming from forcing (i.e.\ complete Heyting/Boolean algebras with a separator reduced to $\{\top\}$) as well as the implicative algebras induced by \emph{combinatory algebras} (CAs) or by \emph{ordered combinatory algebras} (OCAs).
On the other hand, the implicative algebras coming from classical realizability are in general not compatible with joins.
Note that unlike (possibly ordered) combinatory algebras, \emph{partial combinatory algebras} (PCAs) do not induce (full) implicative algebras, but \emph{quasi-implicative algebras}~\cite{miq1}, in which one may have $(\top\rightarrow\top)\neq\top$.
However, as shown in~\cite{miq1}, it is always possible to complete a quasi-implicative algebra into an implicative algebra, simply by adding an extra top element, and without changing the underlying logic.
(Indeed, the triposes associated to a quasi-implicative algebra and to its completion are isomorphic.)
Moreover, when applying this completion mechanism to a quasi-implicative algebra that comes from a PCA, the resulting implicative algebra is always compatible with joins.

In an implicative algebra~$\mathbb{A}$ that is compatible with joins, existential quantification $\lar{\exists}$ may not coincide with supremum $\bigvee$, but both constructions are logically equivalent in the sense that
$$\bigwedge_{(b_{i})_{i\in I}}\!\!
\left(\mathop{\lar{\exists}}\limits_{i\in I}b_i\rightarrow
\bigvee_{i\in I}b_i\right)~\in~\Sigma
\qquad\textrm{and}\qquad
\bigwedge_{(b_{i})_{i\in I}}\!\!
\left(\bigvee_{i\in I}b_i\rightarrow
\mathop{\lar{\exists}}\limits_{i\in I}b_i\right)~\in~\Sigma
\footnote{Note that the first of this two properties actually holds in any implicative algebra.}$$

Now, if we define on $\mathbf{W}$ a new interpretation $\left\|-\right\|^{\J}$ of the language of set theory replacing $\lar{\exists}$ by $\bigvee$ in the definition of the interpretation $\left\|-\right\|$, we easily show (by a straightforward induction) that for all formulas in context $\varphi\,[\underline{x}]$, both denotations $\left\|\varphi\,[\underline{x}]\right\|$ and $\left\|\varphi\,[\underline{x}]\right\|^{\J}$ are equivalent, in the sense that
$$\left\|\varphi\,[\underline{x}]\right\|
~\vdash_{\Sigma[\mathbf{W}^n]}~\left\|\varphi\,[\underline{x}]\right\|^{\J}
\qquad\text{and}\qquad
\left\|\varphi\,[\underline{x}]\right\|^{\J}
~\vdash_{\Sigma[\mathbf{W}^n]}~\left\|\varphi\,[\underline{x}]\right\|\,,$$
where $n$ is the length of $\underline{x}$.
The main interest of the new interpretation is that when the implicative algebra~$\mathbb{A}$ comes from a CA, an OCA, or even a PCA through the completion mechanism mentioned above, the alternative interpretation $\left\|-\right\|^{\J}$ coincides exactly with the Friedman/Rosolini/McCarty realizability interpretation \cite{Friedman,Rosolini,McCarty}. Therefore, intuitionistic realizability models appear to be equivalent to some instances of our construction.

\subsection{The case of classical realizability}
As shown by Krivine~\cite{KRIZF1,KRIZF2}, classical realizability models of $\mathbf{ZF}$ can be constructed from \emph{classical realizability algebras} \cite{KRI11}, or from the (slightly simpler) \emph{abstract Krivine structures} (AKSs) introduced by Streicher in~\cite{STR13}.
Again, both structures are easily reformulated as implicative algebras~\cite{miq1}, which makes possible to compare Krivine's model construction with ours.
Moreover, it has been shown in~\cite{miq1} that every classical implicative algebra is equivalent (from the point of view of the induced triposes) to some Streicher's AKS, which shows that ---at least from a conceptual point of view--- the models arising from classical implicative algebras are essentially the same as the ones arising from classical realizability.
However, the relationship between Krivine's classical realizability models of ZF and our implicative models of ZF is much more intricate than in the intuitionistic case, due to reasons of \emph{polarity} we now need to explain.

For that, let us first recall that in our construction, a name in~$\mathbf{W}$ is a partial function $\alpha\in\mathsf{Part}(\mathbf{W},A)$ associating to each name $\beta\in\partial_0(\alpha)$ (in the domain of~$\alpha$) a truth value $\alpha(\beta)\in A$ that intuitively expresses `how much $\beta$ belongs to $\alpha$'.
So that when $\beta\notin\partial_0(\alpha)$, it is convenient to think that such a truth value implicitly defaults to $\bot$ (that is: `$\beta$ does not belong to~$\alpha$').

As a matter of fact, Krivine's classical realizability interpretation of ZF can be carried out entirely within the same universe~$\mathbf{W}$ as our interpretation---under the hypothesis that the parameterizing implicative algebra~$\mathbb{A}$ is classical, of course.
However, the crucial point is that in Krivine's framework, the elements of~$\mathbf{W}$ have definitely not the same meaning as in ours, since for any two names $\alpha\in\mathbf{W}$ and $\beta\in\partial_0(\alpha)$, the truth value $\alpha(\beta)\in A$ expresses `how much $\beta$ \emph{does not} belong to~$\alpha$' (according to Krivine).
So that when $\beta\notin\partial_0(\alpha)$, such a truth value now implicitly defaults to $\top$.

Formally, Krivine's classical realizability interpretation, written $\left\|-\right\|^{\K}$, takes place in a variant of the language of set theory where the membership predicate $\in$ has been replaced by a \emph{negated membership predicate} $\notin$, from which the usual membership predicate is defined by $x\in y:\equiv\lnot(x\notin y)$.
To each pair of names $\alpha,\beta\in\mathbf{W}$, Krivine associates two truth values $\alpha\notin_{\mathbf{W}}^{\K}\beta$ and $\alpha=_{\mathbf{W}}^{\K}\beta$, that are defined (again) by induction on the ranks of~$\alpha$ and~$\beta$, letting:
$$\begin{array}{r@{~{}~}c@{~{}~}l}
  \alpha\notin_{\mathbf{W}}^{\K}\beta&:=&
  \lar{\forall}_{t\in\partial_0(\beta)}(t=_{\mathbf{W}}^{\K}\alpha\to\beta(t))\\
  \alpha=_{\mathbf{W}}^{\K}\beta&:=&
  (\alpha\subseteq_{\mathbf{W}}^{\K}\beta)\times
  (\beta\subseteq_{\mathbf{W}}^{\K}\alpha),\quad\textrm{where:}\quad
  (\alpha\subseteq_{\mathbf{W}}^{\K}\beta)~:=~
  \lar{\forall}_{t\in\partial_0(\alpha)}(t\notin_{\mathbf{W}}^{\K}\beta\to\alpha(t))\\
\end{array}$$
Notice that Krivine's definition of $\alpha\notin_{\mathbf{W}}^{\K}\beta$ corresponds to the negation of our definition of $\alpha\in_{\mathbf{W}}\beta$, whereas his definition of $\alpha\subseteq_{\mathbf{W}}^{\K}\beta$ is exactly the contraposition of our definition of $\alpha\subseteq_{\mathbf{W}}\beta$---keeping in mind that in Krivine's setting, $\beta(t)$ and $\alpha(t)$ have the same meaning as $\lnot\beta(t)$ and $\lnot\alpha(t)$ in ours%
\footnote{The main benefit of focusing on $\notin$ rather than on~$\in$ is that the recursive interpretations of $\notin$ and $\subseteq$ only rely on universal quantification, whose interpretation is way simpler than existential quantification. The cost of such a design is that it relies on many contrapositions, that require classical reasoning.}.
Once the three primitive relations $\alpha\notin_{\mathbf{W}}^{\K}\beta$, $\alpha\subseteq_{\mathbf{W}}^{\K}\beta$ and $\alpha=_{\mathbf{W}}^{\K}\beta$ have been recursively defined, the usual notion of membership is then recovered letting $\alpha\in_{\mathbf{W}}^{\K}\beta:=\lnot(\alpha\notin_{\mathbf{W}}^{\K}\beta)$, and the rest of the interpretation (written $\left\|-\right\|^{\K}$) is defined the same way as in our framework (cf Section~\ref{s:ImpModels}).

Now if we want to relate Krivine's interpretation with ours, we need to formalize the fact that the same name $\alpha\in\mathbf{W}$ has different meanings according to Krivine and according to us.
For that, we introduce a \emph{set-negation operator} $(\alpha\mapsto\tilde{\alpha}):\mathbf{W}\to\mathbf{W}$ that is defined by induction on the rank of~$\alpha\in\mathbf{W}$, letting:
$$\tilde{\alpha}~:=~
\bigl\{\bigl(\tilde\beta,\lnot\alpha(\beta)\bigr)~:~
\beta\in\partial_0(\alpha)\bigr\}$$
Intuitively, this operator associates to each name $\alpha\in\mathbf{W}$ another name $\tilde{\alpha}\in\mathbf{W}$ that has the same meaning in our framework (resp.\ in Krivine's framework) as~$\alpha$ in Krivine's (resp.\ in ours).
In what follows, it is also convenient to consider set-negation as a unary function symbol as well, that is written and interpreted as the set-negation operator $(\alpha\mapsto\tilde{\alpha}):\mathbf{W}\to\mathbf{W}$.

Using the fact that the parameterizing implicative algebra~$\mathbb{A}$ is classical, we easily check that set-negation is involutive (w.r.t.\ both interpretations), in the sense that:
$$\|\forall x\,(x=\tilde{\tilde{x}})\|~\in~\Sigma
\qquad\textrm{and}\qquad
\|\forall x\,(x=\tilde{\tilde{x}})\|^{\K}~\in~\Sigma\,.$$
We can now prove that both interpretations $\|{-}\|$ and $\|{-}\|^{\K}$ are equivalent \emph{up to set-negation of parameters}, in the sense that for all formulas in context $\varphi\,[\underline{x}]$, we have:
$$\left\|\varphi\,[\underline{\tilde{x}}]\right\|
~\vdash_{\Sigma[\mathbf{W}^n]}~\left\|\varphi\,[\underline{x}]\right\|^{\K}
\qquad\text{and}\qquad
\left\|\varphi\,[\underline{x}]\right\|^{\K}
~\vdash_{\Sigma[\mathbf{W}^n]}~\left\|\varphi\,[\underline{\tilde{x}}]\right\|\,,$$
where $n$ is the length of $\underline{x}$, and writing $\varphi[\underline{\tilde{x}}]$ for $(\varphi[\underline{\tilde{x}}/\underline{x}])[\underline{x}]$.
In particular, for each closed formula~$\varphi$ of ZF, we have:
$$\left\|\varphi\right\|~\vdash_{\Sigma}~\left\|\varphi\right\|^{\K}
\qquad\text{and}\qquad
\left\|\varphi\right\|^{\K}~\vdash_{\Sigma}~\left\|\varphi\right\|\,.$$
So that classical realizability models of ZF are also equivalent to some instances of our construction.

\section{Models of $\mathbf{(I)ZF}$ in a class of toposes}\label{s:ModTopos}
In any elementary topos $\mathcal{E}$ one can interpret first-order languages using the doctrine of subobjects. In particular, an interpretation of the language of $\mathbf{(I)ZF}$ is given by an object $V$ of $\mathcal{E}$ which interprets the universe of sets and by a subobject $\varepsilon$ of $V\times V$ which interprets the membership relation. Equality is always interpreted as the diagonal subobject of $V\times V$. If all axioms of $\mathbf{(I)ZF}$ are validated by the interpretation (that is, if every axiom is interpreted as the maximum subobject) then we get a model of $\mathbf{IZF}$ which is in fact a model of $\mathbf{ZF}$ when $\mathcal{E}$ is boolean.

When the topos $\mathcal{E}$ is obtained as the result of the tripos-to-topos construction from a tripos $\mathsf{P}$, the internal logic of $\mathcal{E}$ can be reduced to the logic of the tripos $\mathsf{P}$ as explained in detail in \cite{VOO08}. Indeed, in this case the objects of $\mathcal{E}$ are pairs $(A,\rho)$ where $A$ is an object of the domain of $\mathsf{P}$ and $\rho\in \mathsf{P}(A\times A)$ is a partial equivalence relation on $A$ with respect to the logic of $\mathsf{P}$, and the subobjects of $(A,\rho)$ correspond to the predicates $\psi\in \mathsf{P}(A)$ which respect the relation $\rho$. Such correspondence extends to a correspondence of connectives and quantifiers between the logics of $\mathsf{P}$ and $\mathcal{E}$.  

The implicative model we produced in the previous section using an implicative algebra $\mathbb{A}$ can be seen as a model of $\mathbf{(I)ZF}$ in the corresponding implicative tripos $\mathsf{P}_{\mathbb{A}}$. Moreover, since $\mathbf{W}$ is a set and $[=_{\mathbf{W}}]$ is an equivalence relation on $\mathbf{W}$ with respect to the logic of $\mathsf{P}_{\mathbb{A}}$ by Lemma \ref{not} (i),(iii),(vi), we have that the pair $(\mathbf{W},[=_{\mathbf{W}}])$ is an object of the topos $\mathbf{Set}[\mathsf{P}_{\mathbb{A}}]$. Finally, as a consequence of Lemma \ref{not} (iv),(v), the relation $[\in_{\mathbf{W}}]$ gives rise to a subobject $\varepsilon_{\mathbf{W}}$ of $(\mathbf{W},[=_{\mathbf{W}}])\times (\mathbf{W},[=_{\mathbf{W}}])$ in $\mathbf{Set}[\mathsf{P}_{\mathbb{A}}]$. By reducing the internal logic of $\mathbf{Set}[\mathsf{P}_{\mathbb{A}}]$ to that of $\mathsf{P}_{\mathbb{A}}$ one obtains that the object $(\mathbf{W},[=_{\mathbf{W}}])$ and the subobject $\varepsilon_{\mathbf{W}}$ define an interpretation of the language of $\mathbf{(I)ZF}$ in $\mathbf{Set}[\mathsf{P}_{\mathbb{A}}]$ which is a model of set theory in there. We thus have proved the following:
\begin{theorem} Every topos $\mathcal{E}$ obtained from an implicative tripos by means of the tripos-to-topos construction from an implicative algebra $\mathbb{A}=(A,\leq,\rightarrow,\Sigma)$ such that $|A|<\kappa$ for some strongly inaccessible cardinal~$\kappa$ hosts a model of $\mathbf{IZF}$. If $\mathbb{A}$ is classical, then $\mathcal{E}$ hosts a model of $\mathbf{ZF}$.
\end{theorem}
Now, as a consequence of Theorem \ref{teomiq}, we obtain the following:
\begin{corollary} If for every cardinal $\kappa'$ there exists a strongly inaccessible cardinal $\kappa$ such that $\kappa'<\kappa$, then every topos obtained from a $\mathbf{Set}$-based tripos by means of the tripos-to-topos construction hosts a model of $\mathbf{IZF}$ (which is a model of $\mathbf{ZF}$ when the topos is boolean). 
\end{corollary}

\subsubsection*{Acknowledgements} The authors would like to acknowledge T. Streicher and F. Ciraulo for useful discussions.

\bibliographystyle{./entics}
\bibliography{biblioPAS}
\end{document}